\documentclass{ieeetj}
\usepackage{cite}
\usepackage{amsmath,amssymb,amsfonts}
\usepackage{graphicx,color}
\usepackage{textcomp}
\usepackage{xcolor}
\usepackage{hyperref}
\hypersetup{hidelinks=true}
\usepackage{algorithm,algorithmic}

\AtBeginDocument{\definecolor{tmlcncolor}{cmyk}{0.93,0.59,0.15,0.02}\definecolor{NavyBlue}{RGB}{0,86,125}}
\usepackage{subcaption}
\usepackage[font=small,labelfont=bf]{caption}
\usepackage{amsmath, amsthm}
\usepackage{siunitx}
\usepackage{booktabs}
\newcommand{\figW}{1}
\usepackage{ulem}
\usepackage{cleveref}
\crefname{figure}{Fig.}{Figs.}
\Crefname{figure}{Fig.}{Figs.}

\newcommand{\invsqrtm}[1]{#1^{-\frac{1}{2}}}

\newtheorem{theorem}{Theorem}

\newcommand{\bbar}[1]{\ensuremath{\boldsymbol{\bar{#1}}}}

\def\diag{\text{diag}}
\def\x{{\mathbf x}}

\def\X{\mathbf{X}}
\def\Y{\mathbf{Y}}
\def\A{\mathbf{A}}
\def\S{\mathbf{S}}

\def\p{\mathbf{p}}
\def\I{\mathbf{I}}

\def\v{\mathbf{v}}
\def\g{\mathbf{g}}

\def\R{\mathbf{R}}
\def\a{\mathbf{a}}

\def\r{\mathbf{r}}

\def\g{\mathbf{g}}

\def\Q{\mathbf{Q}}

\def\C{\mathbf{C}}

\def\m{\mathbf{m}}
\def\n{\mathbf{n}}
\def\s{\mathbf{s}}

\def\y{\mathbf{y}}

\def\min{\text{min}}
\def\max{\text{max}}

\def\argmin{\text{argmin}}

\def\log{\text{log}}

\def\OJlogo{\vspace{-4pt}$<$Society logo(s) and publication title will appear here.$>$}
\def\seclogo{\vspace{10pt}$<$Society logo(s) and publication title will appear here.$>$}

\def\authorrefmark#1{\ensuremath{^{\textbf{#1}}}}

\usepackage{xcolor}
\newcommand{\rev}[1]{\textcolor{red}{#1}}

\begin{document}
\receiveddate{XX Month, XXXX}
\reviseddate{XX Month, XXXX}
\accepteddate{XX Month, XXXX}
\publisheddate{XX Month, XXXX}
\currentdate{21 September, 2025}
\doiinfo{XXXX.2022.1234567}

\markboth{Spatial Power Estimation via Riemannian Covariance Matching}{O. Cohen {et al.}}

\title{Spatial Power Estimation via Riemannian Covariance Matching}

\author{Or Cohen\authorrefmark{1}, Alon Amar\authorrefmark{1}, Ronen Talmon\authorrefmark{1}}
\affil{Andrew and Erna Viterbi Faculty of Electrical and Computer Engineering, Technion, Haifa, Israel}
\corresp{Corresponding author: O. Cohen (email: or.cohen@campus.technion.ac.il).}

\begin{abstract}

We propose a new method for spatial power spectrum estimation in array processing that leverages the Riemannian geometry of Hermitian positive definite (HPD) matrices. We show that conventional approaches minimize variants of the Euclidean distance between the sample covariance matrix and a model covariance matrix, without considering \rev{the fact} that covariance matrices lie on the Riemannian manifold of HPD matrices. By exploiting this manifold, we present a Riemannian-aware covariance matching algorithm, termed SERCOM, using the Jensen–Bregman LogDet (JBLD) divergence, which, unlike other Riemannian distances, can be evaluated efficiently without eigen-decomposition. We theoretically compare the JBLD divergence to other Euclidean- and Riemannian-based distances, demonstrating robustness to spectral distortions. Experimental results \rev{demonstrate} that SERCOM consistently outperforms existing methods in direction-of-arrival (DOA) and power estimation, particularly in challenging scenarios with low SNR, limited number of snapshots, and correlated sources.

\end{abstract} 

\begin{IEEEkeywords}
Beamforming, Riemannian geometry, covariance matching, Jensen–Bregman LogDet, matrix manifold. 
\end{IEEEkeywords}

\maketitle
\section{INTRODUCTION}\label{sec:introduction}
Spatial power estimation is a fundamental problem in array signal processing, with a long history across statistical, subspace, and more recently deep learning–based methods \cite{stoica2005spectral,krim2002two,krishnaveni2013beamforming,pesavento2023three,wu2019deep,liu2018direction,papageorgiou2021deep,mylonakis2024novel}. It is closely connected to direction-of-arrival (DOA) estimation, since DOAs typically appear as dominant peaks in the estimated spectrum. Classical beamformers such as the delay-and-sum (Bartlett's beamformer) and the minimum variance distortionless response (MVDR) \cite{capon2005high,salvati2016weighted} rely on a single-source model, while subspace methods such as multiple signal classification (MUSIC) \cite{schmidt1986multiple,mestre2008modified,yan2013low} estimate DOAs directly without power estimation. 
\rev{Recently, deep learning–based approaches have also been proposed for DOA and spatial spectrum estimation, showing strong empirical performance in specific operating regimes, including low-SNR scenarios, when trained on representative datasets \cite{wu2019deep,liu2018direction,papageorgiou2021deep,mylonakis2024novel}. However, they typically require supervised training with labeled data and may need retraining or adaptation when the array geometry, signal model, or operating conditions differ between the training and deployment stages, which limits their applicability in settings with limited training data or varying sensing conditions.}

\rev{Among model-based methods, greedy sparse recovery approaches are often computationally attractive when the number of sources is small \cite{chen2006theoretical}. However, on fine angular grids the steering dictionary can become highly coherent, which may reduce robustness for closely spaced or correlated sources. Moreover, the greedy selection process is susceptible to noise and small sample sizes, which can cause the algorithm to converge to incorrect directions. Accordingly, many} sparse formulations, including regularization-based methods \cite{chen2006theoretical,yang2015enhancing} and sparse iterative covariance-based estimation (SPICE) algorithms \cite{stoica2010spice,stoica2012spice,stoica2014weighted}, adopt a multi-source model, jointly estimating the entire spatial spectrum.  

Many of these approaches can be recast using the covariance matching framework, where the sample covariance matrix $\widehat{\R}$ is matched to a model covariance matrix $\R$ conveying the expected spatial power spectrum at predefined directions \cite{ottersten1998covariance,li1999computationally,wu2017toeplitz}.
This covariance matching approach formulates the estimation of the spatial power as the minimization of a dissimilarity measure between $\widehat{\R}$ and $\R$, often augmented by regularization terms. 
Specifically, in most widely used formulations this dissimilarity reduces to the Euclidean distance between covariance matrices or to a weighted variant of this distance \cite{stoica2010new,stoica2010spice}.

For example, a common dissimilarity criterion in covariance matching is the asymptotic minimum variance (AMV) \cite{ottersten1998covariance,li1999computationally,abeida2012iterative}, which can be interpreted as a generalized least squares (GLS) fit applied to the vectorized form of the covariance matrix $\R$, treating it as a vector in Euclidean space. Under suitable statistical conditions, minimizing this criterion is equivalent to maximum likelihood estimation of the power spectrum, which in turn coincides with the minimum variance unbiased estimator. Two additional examples are the SPICE algorithm \cite{stoica2010spice,stoica2012spice}, which was introduced as a convex approximation of AMV with an efficient iterative solver, and the iterative sparse asymptotic minimum variance (SAMV) algorithm \cite{abeida2012iterative}, which was proposed for the direct minimization of AMV.

\textbf{\textit{Research gap.}} Although effective in some scenarios, these popular covariance matching formulations treat covariance matrices as vectors in Euclidean space, thereby ignoring their intrinsic geometry as elements of the Riemannian manifold of Hermitian positive definite (HPD) matrices \cite{bhatia2009positive}. 
\rev{Specifically, Euclidean-based methods often experience substantially degraded performance when dealing with ill-conditioned sample covariance matrices or in low-snapshot regimes. By contrast, the Riemannian framework offers a more robust alternative by accounting for the curved geometry of covariance matrices, which inherently regularizes the estimation.}
The HPD manifold admits several Riemannian distances that capture this geometry, including the affine-invariant Riemannian metric (AIRM) \cite{pennec2006riemannian}, which is perhaps the most widely studied, as well as other alternatives developed in the literature. This viewpoint has been successfully exploited in fields such as computer vision, neuroscience, and medical data analysis, consistently outperforming Euclidean approaches \cite{chen2018component,jiang2019sleep,chevallier2021review}.
Following this success, in this paper, we propose a Riemannian geometry–aware covariance matching approach for spatial power spectrum estimation. 

\textbf{\textit{Related Work.}}
Recent works have explored leveraging the Riemannian geometry of covariance matrices for DOA and array processing. One line of research focuses on geometry-aware averaging of covariances, e.g., using the Riemannian (Karcher) mean to suppress interferences while preserving shared spectral components, which can be integrated into classical beamformers or DOA estimators \cite{bar2023interference}. 
Another approach formulates DOA estimation directly via Riemannian distances. Early methods compute pseudo-spectra using the AIRM between the sample covariance and a unit-power single-source model, sometimes extended to include source scaling \cite{coutino2016direction,dong2019scaling}. Related strategies construct per-snapshot HPD covariances, aggregate them with the Riemannian mean, and perform 1D search against single-source models \cite{chahrour2019direction,chahrour2021target}. These methods highlight robustness benefits from geometry-aware matching but remain restricted to single-source models and per-direction searches.
In \cite{wang2019grid}, a \rev{gridless} method for sparse linear arrays was proposed, employing covariance matching with the Bures–Wasserstein distance. Instead of directly estimating the power spectrum, it was combined with root-MUSIC for DOA estimation. Recently, \cite{picard2024riemannian} introduced a covariance-matching beamformer based on the Log-Euclidean (LE) metric \cite{arsigny2006log},\cite{arsigny2007geometric} under the single-source covariance model, which allowed a closed-form expression of the beamformer with respect to a steering vector and improved performance compared to classical, Euclidean-based methods.

\textbf{\textit{Our Contributions.}}
We propose \textbf{SERCOM}, Spatial power Estimation via Riemannian COvariance Matching, a new method that formulates a Riemannian covariance matching task for \rev{estimating} the power spectrum across multiple directions on a predefined grid, from which the DOAs are subsequently identified as the dominant peaks. \rev{Unlike previous HPD geometry-aware approaches, we introduce a multi-source model and apply a scalable optimization by minimizing the \emph{Jensen–Bregman LogDet (JBLD)} divergence \cite{cherian2012jensen,sra2012new} between $\R$ and $\widehat{\R}$. This allows for an efficient Riemannian distance-based estimation that maintains computational tractability even as the number of sensors increases.}
We establish an analytical connection between SPICE, AMV, AIRM, and JBLD criteria, clarifying their similarities and differences in the context of spatial power estimation. The proposed SERCOM method with the JBLD criterion outperforms the SAMV and SPICE algorithms while reducing computational complexity compared to the AIRM, yet preserving key aspects of the intrinsic affine-invariant geometry of covariance matrices. Furthermore, unlike approaches restricted to single-source or unit-power models \cite{coutino2016direction,dong2019scaling,chahrour2019direction,picard2024riemannian}, our formulation is general and incorporates a model of the full power spectrum, offering a geometry-aware alternative to traditional Euclidean-based techniques. Experimental results further show that SERCOM is a computationally practical framework for spatial power estimation, yielding improved source power and direction estimates compared to existing methods in challenging scenarios such as low SNR conditions and correlated sources.

\textbf{\textit{Notations.}}
In this paper, bold symbols such as $\x$ and $\A$ represent vectors and matrices, respectively. The subscript $\x_k$ indicates the $k$-th entry of the vector $\x$. The superscript $(\cdot)^H$ denotes the conjugate transpose. The matrix trace is $\mathrm{tr}(\cdot)$, the symbol $(\cdot)^{-1}$ denotes the inverse of a matrix, while $(\cdot)^{1/2}$ and $(\cdot)^{-1/2}$ on a matrix signify the square root of the matrix and its inverse, respectively. The symbol $\log(\cdot)$ can refer to either the natural scalar logarithm or the matrix logarithm, depending on the context, and $\log\left|\cdot\right|$ is used for the log-determinant operation on a matrix. The Frobenius norm is denoted by $\|\cdot\|_F$, while $\|\cdot\|_2$ denotes the $\ell_2$ norm for vectors and the spectral norm for matrices. The Kronecker product is represented by $\otimes$, and $\oslash$ indicates element-wise division between vectors.

\section{BACKGROUND: RIEMANNIAN MANIFOLD OF HPD MATRICES}\label{sec:background}
Complex $n\times n$ signal-plus-noise covariance matrices lie in the space of HPD matrices, denoted as $\mathcal{P}_n$:
\begin{equation}
    \mathcal{P}_n = \left\{ \R \in \mathbb{C}^{n \times n} \mid \R = \R^H, \R \succ 0 \right\}.
\end{equation}

A widely used Riemannian metric in this space is the AIRM \cite{pennec2006riemannian}, which admits a Riemannian manifold structure with well-defined geometry \cite{bhatia2009positive}.  
This Riemannian geometry admits the logarithmic mapping $\text{Log}_\R: \mathcal{P}_n \to T_\R \mathcal{P}_n$ that maps any $\boldsymbol{\Gamma}\in\mathcal{P}_n$ into the tangent space of $\R$:
\begin{equation}
\label{eq:logmap}
\text{Log}_\R(\boldsymbol{\Gamma}) = \R^{1/2} \log\big(\R^{-1/2}\boldsymbol{\Gamma} \R^{-1/2}\big) \R^{1/2}.
\end{equation}
At any $\R \in \mathcal{P}_n$, the AIRM defines the following inner product in the tangent space $T_\R \mathcal{P}_n$:
\begin{equation}
\label{eq:rie_inner_prod}
\langle \X, \Y \rangle_\R = \mathrm{tr}\Big(\R^{-1} \X \R^{-1} \Y\Big), \quad \X, \Y \in T_\R \mathcal{P}_n.
\end{equation}
This inner product induces a norm of the tangent vectors. In turn, this yields the following Riemannian distance between any two HPD matrices $\R_1, \R_2 \in \mathcal{P}_n$:
\begin{equation}\label{eq:AIRM_Dist}
\mathcal{D}^2_{\mathrm{AIRM}}(\R_1,\R_2) = \big\|\log(\R_1^{-1/2} \R_2 \R_1^{-1/2})\big\|_F^2,
\end{equation}
which is the squared length of the shortest geodesic connecting $\R_1$ and $\R_2$ on the manifold.

Since computing the Riemannian distance induced by the AIRM \eqref{eq:AIRM_Dist} requires matrix square roots and matrix logarithms via 
eigen-decomposition, \rev{which are computationally intensive operations}, several approximations have been proposed to offer more efficient alternatives while approximately preserving the underlying geometry.
One widely used approximation is the LE distance \cite{arsigny2007geometric}:
\begin{equation}\label{eq:LE_Dist}
        \mathcal{D}^2_{\text{LE}}(\R_1,\R_2) = \|\log(\R_1) - \log(\R_2)\|^2_F,
\end{equation}
where the matrix logarithm maps covariance matrices into a Euclidean space, allowing the standard Euclidean distance to be applied. \rev{However, the LE metric requires the computation of matrix logarithms, which incurs significant computational overhead and may consequently hinder the efficiency of downstream tasks.}

Another geometry-aware alternative to the AIRM is the JBLD \cite{cherian2012jensen}. It is obtained by the symmetrized Bergman divergence \cite{banerjee2005clustering}\cite{nielsen2007centroids}, and defined by
\begin{equation}
        \label{eq:JBLD}
        \mathcal{D}^2_{\mathrm{JBLD}}(\R_1,\R_2) = \log\left|\frac{\R_1 + \R_2}{2}\right| - \frac{1}{2}\log|\R_1 \R_2|.
\end{equation}
It was proposed in \cite{cherian2012jensen} for fast nearest-neighbor retrieval in a \rev{dataset of covariance matrices} and was shown to outperform LE and other geometry-aware distances in applications involving covariance similarity measures. In \cite{sra2012new}, it was shown that the square root of the divergence $\mathcal{D}_{\mathrm{JBLD}}$ is a distance, with tight links to the AIRM and its desirable properties.
Moreover, unlike AIRM and LE, JBLD avoids the computationally intensive matrix logarithm for distance and gradient evaluation, and can be efficiently and stably computed \rev{via} Cholesky factorization \cite{cherian2012jensen}.

\section{PROBLEM FORMULATION AND EXISTING METHODS} \label{sec:problem_formulation}
Consider an $M$-sensor array with arbitrary but known configuration, receiving $K$ far-field, narrowband sources at a common carrier frequency, impinging from directions $\phi_1,\ldots,\phi_K$ with respective powers $\sigma_1^2,\ldots,\sigma_K^2$. Given $N$ snapshots collected by the array, the goal is to accurately estimate the directions and powers of the sources, while the number of sources $K$ is unknown. The snapshot of the received signals $\y(t) \in \mathbb{C}^M$ at all sensors is:
\begin{equation}
    \begin{aligned}
        &\y(t) = \sum_{k=1}^{K}{\a_{\phi_k} s_k(t) + \n(t)} \\
        &= \A_{\boldsymbol{\phi}}\s(t) + \n(t), \;\; t = 1,...,N \\
    \end{aligned}
\end{equation}
where $\A_{\boldsymbol{\phi}} = [\a_{\phi_1},\ldots,\a_{\phi_K}] \in \mathbb{C}^{M\times K}$ is the steering matrix, with each column being a steering vector $\a_{\phi_k}$, a known function of $\phi_k$ given by the carrier frequency and the array configuration. The vector $\s(t) = [s_1(t),\ldots,s_K(t)]^T \in \mathbb{C}^{K}$ is the concatenation of all incident signals, where $s_k(t)$ is the complex signal emitted by the $k$-th source with zero mean and power $\sigma_k^2$. The additive noise $\n(t) \sim \mathcal{CN}(\boldsymbol{0},\sigma_n^2\I)$ has a known power $\sigma_n^2$. It is assumed that $\s(t)$ and $\n(t)$ are independent and that $\mathbb{E}[\s(t)\s(\tilde{t})^H] = \delta_{t,\tilde{t}}\boldsymbol{\Sigma}$, where $\boldsymbol{\Sigma} = \diag(\sigma_1^2,...,\sigma_K^2)$, where $\delta_{t,\tilde{t}}$ is 1 only if $t = \tilde{t}$ and 0 otherwise.

A common line of direction and power estimation methods, termed 
on-grid methods \cite{stoica2005spectral}, defines an angular grid with $D$ examined directions, $\boldsymbol{\theta} = [\theta_1,\ldots,\theta_D]^T$ and estimates the power at each grid point. This can be formulated as a covariance matching problem \cite{ottersten1998covariance,stoica2010spice}, where the population covariance $\R=\mathbb{E}[\y(t)\y(t)^H]$ is parameterized by the spatial power spectrum and matched to the sample covariance $\widehat{\R} = \tfrac{1}{N}\sum_{t=1}^N \y(t)\y(t)^H$, under the standing assumption that $N \ge M$, so that $\widehat\R$ is HPD.

Specifically, the on-grid \emph{multi-source} covariance model incorporates all steering vectors of the angular grid:
\begin{equation}\label{eq:covariance_model}
    \R(\p;\boldsymbol{\theta})
    = \A_{\boldsymbol{\theta}}\diag(\p)\A_{\boldsymbol{\theta}}^H
      + \sigma_n^2\I,
\end{equation}
where $\p$ is the spatial power spectrum over the grid, denoted simply as $\R(\p)$ hereafter. Then, the following global optimization problem is solved to jointly estimate all grid powers:
\begin{equation}\label{eq:cov_mat}
        \widehat{\p} = \underset{\p}{\mathrm{\argmin}} \; \mathcal{D}^2(\R(\p),
        \widehat{\R}),\;\;\; \p \in \mathbb{R}^D_+,
\end{equation}
for a chosen distance measure $\mathcal{D}$. \rev{Methods that solve \eqref{eq:cov_mat} are agnostic to the sensor arrangement: it enters only through the covariance model in \eqref{eq:covariance_model}, which can be defined for ULA/UCA/UPA or any other known array configuration.}

The linear relation between the covariance matrix model $\R(\p)$ and the grid power vector $\p$ led to the derivation of the AMV criterion \cite{ottersten1998covariance},\cite{abeida2012iterative}:
\begin{align}\label{eq:AMV}
    \mathcal{D}^2_{\mathrm{AMV}}(\R(\p), \widehat{\R}) 
    &= (\hat{\r} - \r)^H \C_\r^{-1} (\hat{\r} - \r) \nonumber \\
    &= \|\invsqrtm{\R}(\widehat{\R} - \R)\invsqrtm{\R}\|_F^2,
\end{align}
where $\C_\r \triangleq \R^H \otimes \R$, $\r \triangleq \Vec{(\R)}$, and $\hat{\r} \triangleq \Vec{(\widehat{\R})}$. This criterion depends non-linearly on the power vector $\p$, and it is not convex. Importantly, it led to the following criterion used in the well-known SPICE algorithm \cite{stoica2010spice}:
\begin{align}\label{eq:SPICE}
    \mathcal{D}^2_{\mathrm{SPICE}}(\R(\p),
        \widehat{\R}) = \|\invsqrtm{\R}(\widehat{\R} - \R)\invsqrtm{\widehat{\R}}\|_F^2.
\end{align}
This is a convex objective\cite{stoica2010spice}, which can be efficiently optimized using an iterative process. In \cite{abeida2012iterative}, another iterative algorithm, called SAMV, was suggested, directly minimizing the AMV criterion while maintaining a sparse solution. 

In the context of this paper, we emphasize that both SAMV and SPICE minimize a weighted \emph{Euclidean distance} between covariance matrices. However, they do not take into account the geometry of covariance matrices as points on the Riemannian manifold of HPD matrices.

\section{PROPOSED ALGORITHM}\label{sec:proposed_approach}
Following the on-grid covariance matching approach with a multi-source model \eqref{eq:cov_mat}, we propose to estimate the source power spectrum by minimizing the Jensen–Bregman LogDet divergence~\eqref{eq:JBLD} between the model and sample covariance matrices:

\begin{equation}
        \label{eq:SERCOM}
        \hat{\p} = \underset{\p \in \mathbb{R}^D_+}{\arg\min} \;\;
        \log\left|\frac{\widehat\R + \R(\p)}{2}\right| - \frac{1}{2}\log|\R(\p)|,
\end{equation}
where the constant term $-\tfrac{1}{2}\log|\widehat\R|$ is omitted since it does not depend on $\p$.
This covariance matching criterion considers the HPD Riemannian geometry of the covariance matrices, in contrast to the ``Euclidean'' AMV, and SPICE criteria \eqref{eq:AMV}-\eqref{eq:SPICE}. 
Moreover, unlike the AIRM distance \eqref{eq:AIRM_Dist}, the JBLD-based formulation in \eqref{eq:SERCOM} remains valid even when $N < M$, since $\widehat\R$ does not need to be a full rank matrix.

We will show that it leads to superior performance while maintaining similar or better computational demand.

The optimization of the proposed estimator \eqref{eq:SERCOM}, termed SERCOM, is outlined in Algorithm~\ref{alg:JBLD_power_estimation}.
For initialization, we set $\p^{(0)}$ using the classical delay-and-sum beamformer:
\begin{equation}
    \p^{(0)}_d = \a_{\theta_d}^H\widehat{\R}\a_{\theta_d}, \quad d=1,\ldots,D.
\end{equation}
The power vector $\p$ is then updated iteratively in a gradient-based manner, where the gradient at each iteration is given by
\begin{equation}
\nabla_{p_d} \mathcal{D}_{\mathrm{JBLD}}^2(\R(\p), \widehat{\R})  = \a_{\theta_d}^H\left(\nabla_\R \mathcal{D}_{\mathrm{JBLD}}^2(\R(\p), \widehat{\R}) \right)\a_{\theta_d}
\label{eq:gradient_p_JBLD}
\end{equation}
where
\begin{equation}
\nabla_{\R} \mathcal{D}_{\mathrm{JBLD}}^2(\R, \widehat{\R}) = (\R + \widehat{\R})^{-1} -\frac{1}{2}\R^{-1}
\label{eq:gradient_R_JBLD}
\end{equation}
The \rev{update of the} vector $\p$ is carried out using the \emph{Adam} optimizer \cite{kingma2014adam}, which combines momentum \cite{nesterov1986} and adaptive learning rates \cite{duchi2011adaptive}. The momentum stabilizes updates for convergence in direction of consistent descent, while the adaptive learning rates \rev{use} the squared gradients for each component of $\p$ to ensure that coordinates with large variability are dampened. This coordinate-wise normalization is well-suited for sparse spectrum estimation, where only a small subset of entries in $\p$ are expected to be nonzero, corresponding to active sources, and different coordinates require different scaling. After each Adam update, the power vector is projected onto the non-negative orthant by enforcing $p_d = \max(0, p_d)$, consistent with its interpretation as a power spectrum. Together, momentum, adaptive normalization, and non-negativity projection yield a stable optimization trajectory that favors accurate and sparse recovery of the sources. The iterations continue until the relative change in $\p$ between consecutive iterations falls below a predefined threshold, $\epsilon_p$, ensuring convergence to a stable estimate.

\rev{
While second-order strategies are natural candidates for log-determinant objectives, our preliminary experiments with MM-style and Newton-type updates were substantially slower per iteration and did not yield improved solutions. For this reason, we focus on the Adam-based implementation in the reported results.
}

\rev{
The objective in \eqref{eq:SERCOM} is generally non-convex in the vector $\p$. Consequently, SERCOM is not guaranteed to converge to the global optimum, and may in principle converge to suboptimal stationary points. In practice, we observed stable convergence with the Adam updates and low sensitivity to initialization across several reasonable choices, including delay-and-sum, MVDR, and a Log-Euclidean–based initialization \cite{picard2024riemannian}.
}
\begin{algorithm}[h!] 
\caption{SERCOM using JBLD}
\begin{algorithmic}[1] \label{alg:JBLD_power_estimation}
\STATE \textbf{Input:} $\widehat{\R}$, $\sigma_n^2$, $\boldsymbol{\theta}\in\mathbb{R}^D$,  $\eta$, $\beta_1$, $\beta_2$, $\epsilon_v$ \textit{maxiter}, $\epsilon_p$ 
\STATE $\p^{(0)} = [\a_{\theta_d}^H\widehat{\R}\a_{\theta_d}]_{d=1}^D$
\STATE $m^{(0)} = 0$, $v^{(0)} = 0$
\FOR{$i=1$ to \textit{maxiter}}
    \STATE construct $\R(\p^{(i-1)}; \boldsymbol{\theta})$\;\textit{// using ~\eqref{eq:covariance_model}}
    \STATE $\g^{(i)} = \nabla_{\p} \mathcal{D}_{\mathrm{JBLD}}^2(\R(\p^{(i-1)}; \boldsymbol{\theta}), \widehat{\R}) $\;\textit{// using ~\eqref{eq:gradient_p_JBLD},\eqref{eq:gradient_R_JBLD}}
    \STATE $\m^{(i)} = \beta_1 \m^{(i-1)} + (1-\beta_1) \g^{(i)}$
    \STATE $\v^{(i)} = \beta_2 \v^{(i-1)} + (1-\beta_2) (\g^{(i)})^2$
    \STATE $\hat{\g}^{(i)} = \left(\m^{(i)}/(1-\beta_1^i)\right) \oslash \left(\sqrt{\v^{(i)}/(1-\beta_2^i)} + \epsilon_v \right)$
    \STATE $\p^{(i)} = \p^{(i-1)} - \eta \hat{\g}^{(i)}$
    \STATE $\p^{(i)} = \max(\p^{(i)}, 0)$
    \IF{$\| \p^{(i)} - \p^{(i-1)} \| / \| \p^{(i-1)} \| < \epsilon_p $} 
        \STATE \textbf{break}
    \ENDIF
\ENDFOR
\STATE \textbf{Output:} $\p^{(t)}$
\end{algorithmic}
\end{algorithm}

\rev{
The hyperparameter settings for Algorithm~\ref{alg:JBLD_power_estimation} are detailed in Table~\ref{tab:hyperparameters}.
}
\section{THEORETICAL ANALYSIS OF THE MATCHING CRITERIA}

In this section, we theoretically examine the AMV, SPICE, AIRM, and JBLD as criteria for matching a model covariance\footnote{In this section, we write $\R(\p)$ simply as $\R$ for brevity.}
$\R(\p)$ to the sample covariance $\widehat\R$, obtained from $N$ snapshots, and establish the conditions under which the criteria coincide and under which they diverge.

We begin by demonstrating that these criteria share a common mathematical structure, even though they arise from different perspectives (Euclidean and Riemannian).
To this end, we define the HPD matrix\footnote{For this theoretical analysis section, we assume $N \ge M$ such that $\widehat\R$ is surely HPD}
\begin{equation} \label{eq:Q}
\Q \triangleq \R^{-1/2}\,\widehat\R\,\R^{-1/2} \;\in\; \mathcal{P}_M.
\end{equation}
Then, the AIRM distance $\|\log(\Q)\|_F^2$ \eqref{eq:AIRM_Dist} reduces, after linearizing the matrix logarithm near the identity matrix $\I$,
\[
\log(\Q) \approx \Q - \I = \R^{-1/2}(\widehat{\R}-\R)\R^{-1/2},
\]
to the AMV criterion \eqref{eq:AMV}.
Thus, AMV may be interpreted as a first-order approximation of AIRM around $\Q=\I$ (i.e., when $\R=\widehat\R$). From a manifold viewpoint, the AIRM distance measures the Riemannian norm \eqref{eq:rie_inner_prod} of the logarithmic map $\mathrm{Log}_{\R}(\widehat\R)$ \eqref{eq:logmap}, namely, the mapping of $\widehat{\R}$ to the tangent space at $\R$, which is often viewed as the Riemannian counterpart of Euclidean subtraction. The AMV replaces this Riemannian operation simply by the Euclidean subtraction $\widehat\R-\R$.

Given the known proximity of JBLD to AIRM \cite{cherian2012jensen,sra2012new}, and the close relationship between SPICE and AMV, it is natural to expect that when $\R$ and $\widehat\R$ are close, all four criteria exhibit similar behavior.

For formal analysis, we express the closeness of $\R$ and $\widehat\R$ in terms of two main factors. First, the \emph{model mismatch}, i.e., how well $\R$ represents the population covariance $\bbar\R\in\mathcal{P}_M$ that generates the data $\{ \y(t) \}_{t=1}^N \overset{\mathrm{i.i.d.}}{\sim} \mathcal{CN}(0,\bar\R)$. This model mismatch often depends on factors such as array calibration imperfections and correlated or off-grid sources. Second, the \emph{sampling error}, i.e., the deviation of the sample covariance $\widehat\R$ from the population covariance $\bbar\R$, due to the finite number of snapshots $N$.
Using these two factors, the proximity \rev{between} $\R$ and $\widehat\R$ is bounded by:
\begin{equation} \label{eq:model_plus_sampling}
\|\widehat{\R}\R^{-1}-\I\|_2 
\;\le\; \|\bbar\R\R^{-1}-\I\|_2
\;+\;\|(\widehat\R-\bbar\R)\R^{-1}\|_2 \;\triangleq\; \varepsilon,
\end{equation}
where $\varepsilon$ consists of two parts: the first term represents the model mismatch, 
and the second term represents the sampling error. Here $\|\cdot\|_2$ denotes the spectral norm, 
i.e., the largest singular value.

\subsection{Asymptotic Equivalence of Criteria}

Define the AIRM-ball of radius $\rho>0$ around $\bar\R$ by
\begin{equation}
\mathcal B_{\mathrm{AIRM}}(\bar\R,\rho) 
=\left\{\,\R\in\mathcal{P}_M \;\middle|\;
\mathcal{D}_{\mathrm{AIRM}}(\bar\R,\R)<\rho \right\},
\end{equation}
and the condition number $\kappa(\bar\R)=\|\bar\R\|_2\cdot\|\bar\R^{-1}\|_2$.

The following result shows that with sufficiently small model mismatch and a large enough number of samples, all four criteria become asymptotically close in a neighborhood of $\widehat\R$. In particular, their differences vanish at different orders of $\varepsilon$ defined in \eqref{eq:model_plus_sampling}, all relative to the SPICE criterion, which serves as a convenient baseline.

\begin{theorem}[Asymptotic equivalence of criteria]
\label{thm:asymptotic_closenseness}
Let $\widehat \R$ be the sample covariance from $N$ snapshots.
Given a target $\varepsilon\in(0,1)$, confidence $1-\delta$, and radius 
$\rho \in (0,\log(1+\varepsilon))$, if
\begin{equation}
N \;\ge\; N_0(\varepsilon,\delta,\rho)
:= \frac{\kappa(\bar\R)^2\,e^{2\rho}}
        {\big(\varepsilon-(e^\rho-1)\big)^2}
   \,\Big(M+\log(2/\delta)\Big),
\end{equation}
then with probability at least $1-\delta$, for every 
$\R\in\mathcal B_{\mathrm{AIRM}}(\bar\R,\rho)$ we have
\begin{equation}
\|\widehat\R\R^{-1}-\I\|_2 \le \varepsilon,
\end{equation}
and
\begin{equation}
\begin{aligned}
\big|\mathcal{D}_{\mathrm{AMV}}^2(\R,\widehat\R)
   -\mathcal{D}_{\mathrm{SPICE}}^2(\R,\widehat\R)\big|
   &\le M\cdot C_{\mathrm{AMV}}(\varepsilon)\,\varepsilon^3, \\
\big|\mathcal{D}_{\mathrm{AIRM}}^2(\R,\widehat\R)
   -\mathcal{D}_{\mathrm{SPICE}}^2(\R,\widehat\R)\big|
   &\le M\cdot C_{\mathrm{AIRM}}(\varepsilon)\,\varepsilon^4, \\
\big|8\mathcal{D}_{\mathrm{JBLD}}^2(\R,\widehat\R)
   -\mathcal{D}_{\mathrm{SPICE}}^2(\R,\widehat\R)\big|
   &\le M\cdot C_{\mathrm{JBLD}}(\varepsilon)\,\varepsilon^4,
\end{aligned}
\end{equation}
where $C_{\mathrm{AMV}},C_{\mathrm{AIRM}},C_{\mathrm{JBLD}}$ are functions 
of $\varepsilon$ that converge to finite constants as $\varepsilon\to 0$.

In particular, as $N\to\infty$ and $\varepsilon\to 0$, all criteria become equivalent in a neighborhood of the shared minimizer $\R=\widehat\R$.
\end{theorem}

The proof is provided in Appendix~\ref{appendix:proof_theorem1}.
Theorem~\ref{thm:asymptotic_closenseness} states that in the absence of model mismatch, i.e., when $\R$ can approach $\bbar\R$ arbitrarily closely, the four criteria \emph{asymptotically converge to one another} locally around the shared minimizer $\R=\widehat\R$. AMV coincides with SPICE up to 
$\mathcal{O}(\varepsilon^3)$, while AIRM and the scaled JBLD coincide up 
to $\mathcal{O}(\varepsilon^4)$. Outside this benign regime, e.g., when we have fewer snapshots, low SNR, correlated or closely spaced sources, etc., the criteria can diverge; nevertheless, our simulations demonstrate that the geometry-aware AIRM and JBLD retain empirical accuracy in these challenging conditions.

\subsection{Criteria Robustness}

All examined criteria can be expressed as a sum of $M$ scalar functions of the eigenvalues of $\Q$ \eqref{eq:Q}, $\{ \lambda_m \}_{m=1}^M$: 
\begin{equation} \label{eq:sum_of_psi}
\mathcal D_\bullet^2(\R,\widehat\R)
\;=\;\sum_{m=1}^M \psi_\bullet(\lambda_m),
\end{equation}
where
\begin{equation}
\begin{aligned}
\label{eq:psi_functions}
\psi_{\mathrm{AMV}}(\lambda)&=(\lambda-1)^2, 
& \psi_{\mathrm{SPICE}}(\lambda)&=\Big(\sqrt{\lambda}-\tfrac{1}{\sqrt{\lambda}}\Big)^2,\\
\psi_{\mathrm{AIRM}}(\lambda)&=(\log \lambda)^2, 
& \psi_{\mathrm{JBLD}}(\lambda)&=\log\!\Big(\tfrac{1+\lambda}{2\sqrt{\lambda}}\Big).
\end{aligned}
\end{equation}
See derivation in Appendix~\ref{appendix:derivation_psi_functions}. Figure~\ref{fig:psi_functions} illustrates these functions $\psi_\bullet(\lambda)$.

\begin{figure}[b!]
    \centering
    \includegraphics[width=0.9\figW\columnwidth]{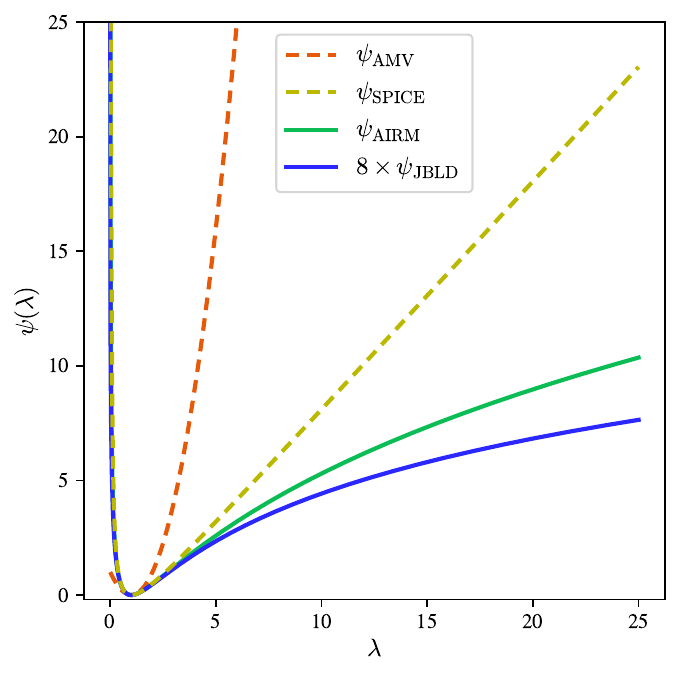}
    \caption{The $\psi(\lambda)$ function for the different covariance matching criteria. $\psi_{\mathrm{JBLD}}$ is scaled to ease visualization.}
    \label{fig:psi_functions}
\end{figure}

Expressing the different criteria in this unified form allows us to examine their robustness.
By the definition of $\Q$ \eqref{eq:Q}, when $\widehat\R$ and $\R$ are close, $\Q$ reduces to the identity matrix and its eigenvalues $\{ \lambda_m \}_{m=1}^M$ equal 1. However, in practice, abrupt noise, limited snapshots, or other distortions of $\widehat\R$ can give rise to ``outlier'' eigenvalues in $\Q$, deviating significantly from $1$. While AMV penalizes deviations from $\lambda=1$ quadratically (see Figure~\ref{fig:psi_functions}), SPICE grows roughly linearly for large $\lambda$. In contrast, AIRM and JBLD increase logarithmically, and are therefore expected to be less sensitive when an eigenvalue of $\Q$ becomes substantially larger than 1. The following theorem formalizes this observation.

\begin{theorem}[Criteria robustness to eigenvalue deviations]
\label{thm:criteria_robustness} 
Suppose that for some $\varepsilon \in (0,1)$ all but one of the eigenvalues of $\Q$ remain close to $1$, i.e. $|\lambda_m - 1| \leq \varepsilon$ for all $m \neq r$, while a single eigenvalue can be expressed as $\lambda_r = 1 + \Delta$, with $\Delta$ dominant in logarithmic scale:
\[
\log(1+\Delta)\ \ge\ 
\max\{\log(1+\varepsilon),\,-\log(1-\varepsilon)\}.
\]
Define the relative contribution of $\lambda_r$ under criterion $\bullet$ as
\[
s_r^{(\bullet)} = 
\frac{\psi_{\bullet}(\lambda_r)}{\sum_{m=1}^M \psi_{\bullet}(\lambda_m)},
\]
where the different $\psi_\bullet$ are defined in~\eqref{eq:psi_functions}. Then
\[
s_r^{(\mathrm{AMV})} \;\ge\; 
s_r^{(\mathrm{SPICE})} \;\ge\; 
s_r^{(\mathrm{AIRM})} \;\ge\; 
s_r^{(\mathrm{JBLD})}.
\]
\end{theorem}

The proof is provided in Appendix~\ref{appendix:proof_theorem2}.
Theorem \ref{thm:criteria_robustness} states that when a single eigenvalue deviates strongly from $1$, AMV assigns 
the highest relative weight to the outlier, followed by SPICE, AIRM, and 
JBLD in decreasing sensitivity.
In spatial power estimation and DOA, large eigenvalue deviations may arise from spurious interference or multi-path propagation. Thus, Theorem~\ref{thm:criteria_robustness} implies that in such cases the AIRM- and JBLD-based methods may help mitigate such effects, yielding more robust covariance matching compared to the AMV and SPICE criteria.

\section{EXPERIMENTAL RESULTS}
\rev{
We demonstrate the performance of Algorithm~\ref{alg:JBLD_power_estimation}, referred to as SERCOM(JBLD), in simulations. We compare it with methods that estimate both source directions and powers through multi-source spatial power spectrum estimation, namely SPICE \cite{stoica2010spice,stoica2012spice,stoica2014weighted} and SAMV \cite{abeida2012iterative}. Additionally, we include covariance-matching baselines based on the AIRM and LE metrics, denoted SERCOM(AIRM) and SERCOM(LE), respectively, implemented analogously to Algorithm~\ref{alg:JBLD_power_estimation} but using the gradient of the cost functions in \eqref{eq:AIRM_Dist} and \eqref{eq:LE_Dist}. Table~\ref{tab:hyperparameters} lists the hyperparameters used for all SERCOM variants.
}

\begin{table}[h!]
    \centering
    \caption{Optimization hyperparameters for SERCOM}
    \label{tab:hyperparameters}
    \begin{tabular}{ccl}
        \toprule
        \textbf{Name} & \textbf{Value} & \textbf{Description} \\
        $\eta$ & $10^{-2}$ & Step size. \\
        $\beta_1$ & 0.9 & Exponential decay rate for the first moment. \\
        $\beta_2$ & 0.999 & Exponential decay rate for the second moment. \\
        $\epsilon_v$ & $10^{-8}$ & Constant added for numerical stability. \\
        \textit{maxiter} & 5000 & Maximum number of iterations. \\
        $\epsilon_p$ & $10^{-4}$ & Tolerance for power spectrum relative change. \\
        \bottomrule
    \end{tabular}
\end{table}

\rev{
For SERCOM(LE), we found it necessary to add a small ``ramp'' diagonal loading to the modeled covariance $\R$ to stabilize the evaluation of the matrix-logarithm gradient with respect to $\p$ (mitigating numerical issues when $\R$ has nearly repeated eigenvalues). This extra stabilization step is a practical drawback of optimizing the LE objective compared to SERCOM(JBLD).
}
\rev{
Unless stated otherwise, the angular grid is $0^\circ$--$180^\circ$ with resolution $0.5^\circ$ (thus $D=361$). 
}

\rev{
The signal-to-noise ratio (SNR) is defined as:
\begin{equation}
    \text{SNR} = 10 \log_{10} \left( \frac{\max_{k} \sigma_k^2}{\sigma_n^2} \right) \; \text{dB},
\end{equation}
where $\sigma_n^2$ denotes the noise power.
}
\rev{
Given an estimated spatial spectrum $\hat{\p}$, we extract DOA estimates by selecting the $K$ most prominent peaks. We report Monte-Carlo results with $L=500$ independent trials. We evaluate
\begin{equation}
\mathrm{RMSE}_{\mathrm{Power}} = \sqrt{\tfrac{1}{LK} \sum_{l=1}^{L} \sum_{k=1}^K \left(\sigma_k^2 - \widehat{\sigma}_k^2\right)^2},
\end{equation}
where $\sigma_k^2$ are the true source powers, and $\widehat{\sigma}k^2$ denotes the power associated with the $k$-th selected peak of the estimated spatial spectrum. The RMSE of the DOA estimates is
\begin{equation}
\mathrm{RMSE}_{\mathrm{DOA}} = \sqrt{\tfrac{1}{LK} \sum_{l=1}^{L} \sum_{k=1}^K \left(\phi_k - \hat{\phi}_k\right)^2},
\end{equation}
where $\phi_k$ are the true DOAs, and $\hat{\phi}_k$ denotes the grid direction corresponding to the $k$-th selected peak. 
}
\rev{
In \cref{fig:rmse_snr_ula,fig:rmse_n_ula,fig:rmse_dtheta,fig:rmse_rho,fig:rmse_snr_uca}, the shaded band indicates the interquartile range (25th–75th percentiles) of the per-trial squared error after applying the square root. For reference, the figures reporting RMSE in direction estimation also depict the stochastic Cramér–Rao bound (CRB) of angle estimation under the corresponding DOA scenarios \cite{stoica2002music}.
}

\rev{
\textbf{\textit{RMSE versus SNR.}}
We consider a ULA with half-wavelength spacing and $M=12$ sensors. In the first experiment, we evaluate the estimation performance for various SNR levels, ranging from $-4.5$ to $4.5$~dB. Three uncorrelated sources are located at on-grid directions $\theta_1 = \ang{35}$, $\theta_2 = \ang{43}$, and $\theta_3 = \ang{51}$. The power of the first two sources is $0$~dB, and the power of the third source is $-5$~dB. The number of snapshots is $N = 50$.
}
\rev{
\begin{figure}[htbp]
    \centering
    \includegraphics[width=\figW\columnwidth]{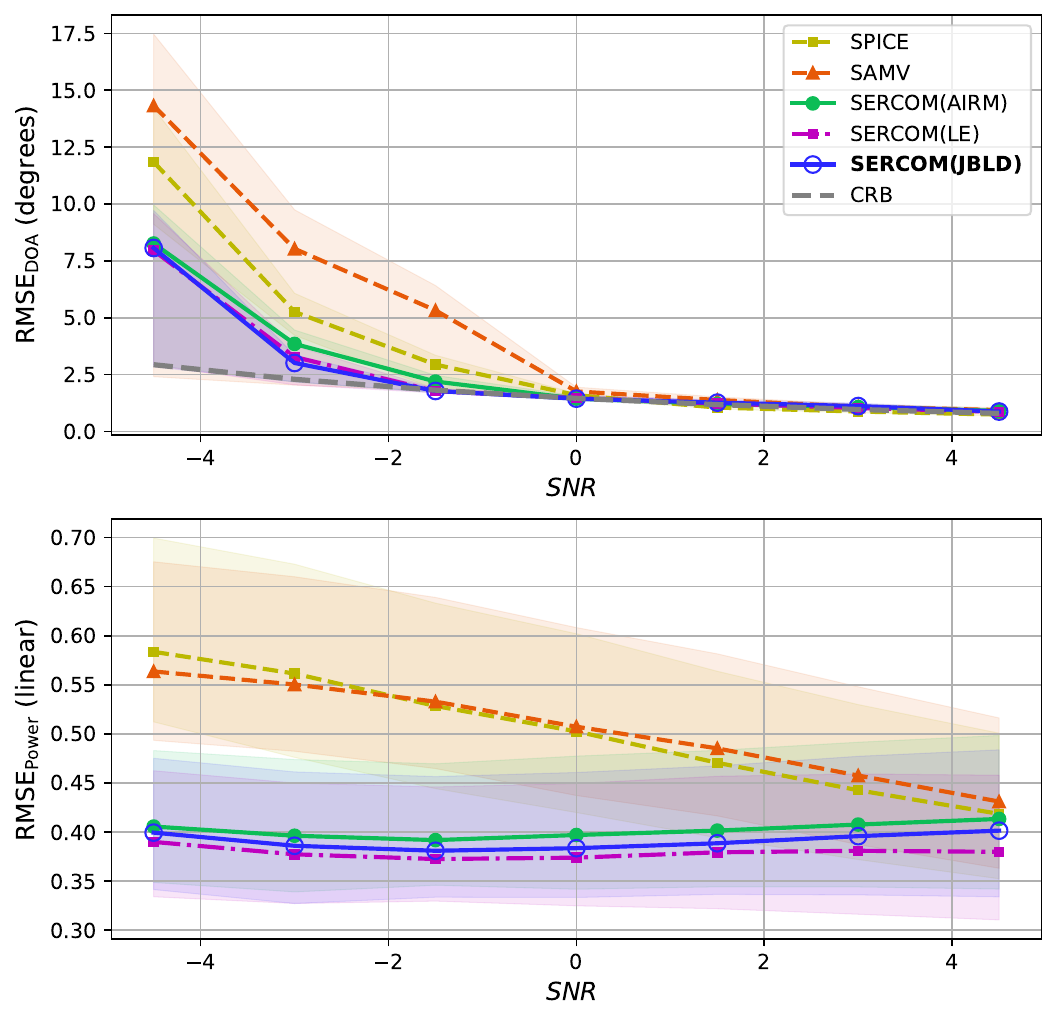}
    \caption{Algorithms' DOA (top) and power (bottom) estimation error at different SNR values.}
    \label{fig:rmse_snr_ula}
\end{figure}
}
\rev{
Figure~\ref{fig:rmse_snr_ula} reports the RMSE versus SNR. The top panel shows DOA estimation errors as a function of SNR. All SERCOM variants exhibit clear improvements over the Euclidean baselines. In particular, SERCOM(JBLD) achieves low errors at low SNR, approaching the CRB earlier than competing methods. The SERCOM(AIRM) and SERCOM(LE) algorithms perform similarly, providing a noticeable gain over SPICE and SAMV in the low-SNR regime.
}
\rev{
The bottom panel shows the corresponding RMSE in power estimation. The Riemannian approaches again outperform the Euclidean ones across all SNR values. All SERCOM variants yield consistently accurate power estimates, while SPICE and SAMV yield noticeably larger power-estimation errors. We attribute this to their broader lobes in the estimated spectrum, which spread the source power across neighboring angles. An illustrative example is provided in Figure~\ref{fig:power_spectrum_example}, showing the estimated spatial power spectrum, $\hat\p$, at $\text{SNR}=-1.5$ dB. These results highlight the robustness of the proposed Riemannian matching approach for power recovery even under low-SNR conditions.
}
\begin{figure}[ht]
    \centering
    \includegraphics[width=\figW\columnwidth]{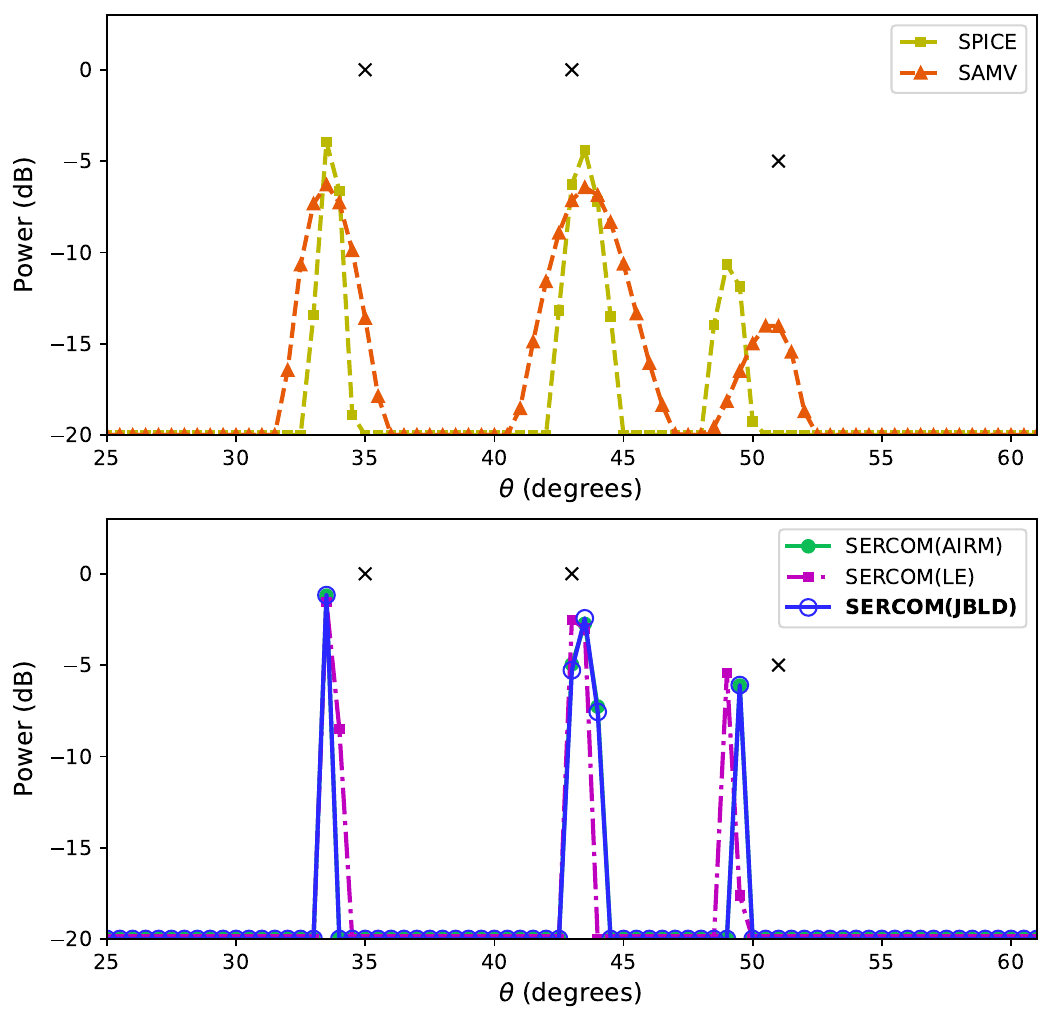}
    \caption{Spatial power spectrum estimated by the different algorithms. The DOA setting is as in the first experiment. The 'x's are located at the sources’ true direction and power.}
    \label{fig:power_spectrum_example}
\end{figure}

\rev{
\textbf{\textit{RMSE versus number of snapshots.}}
Next, we study the effect of the number of snapshots, $N$, on the performance. We fix the ULA configuration and sources of the previous experiment, with $\text{SNR}=-1.5$\,dB, and varied $N$ relative to the number of sensors $M=12$, from $N=M$ up to $N=7M$.
}
\begin{figure}[htbp]
    \centering
    \includegraphics[width=\figW\columnwidth]{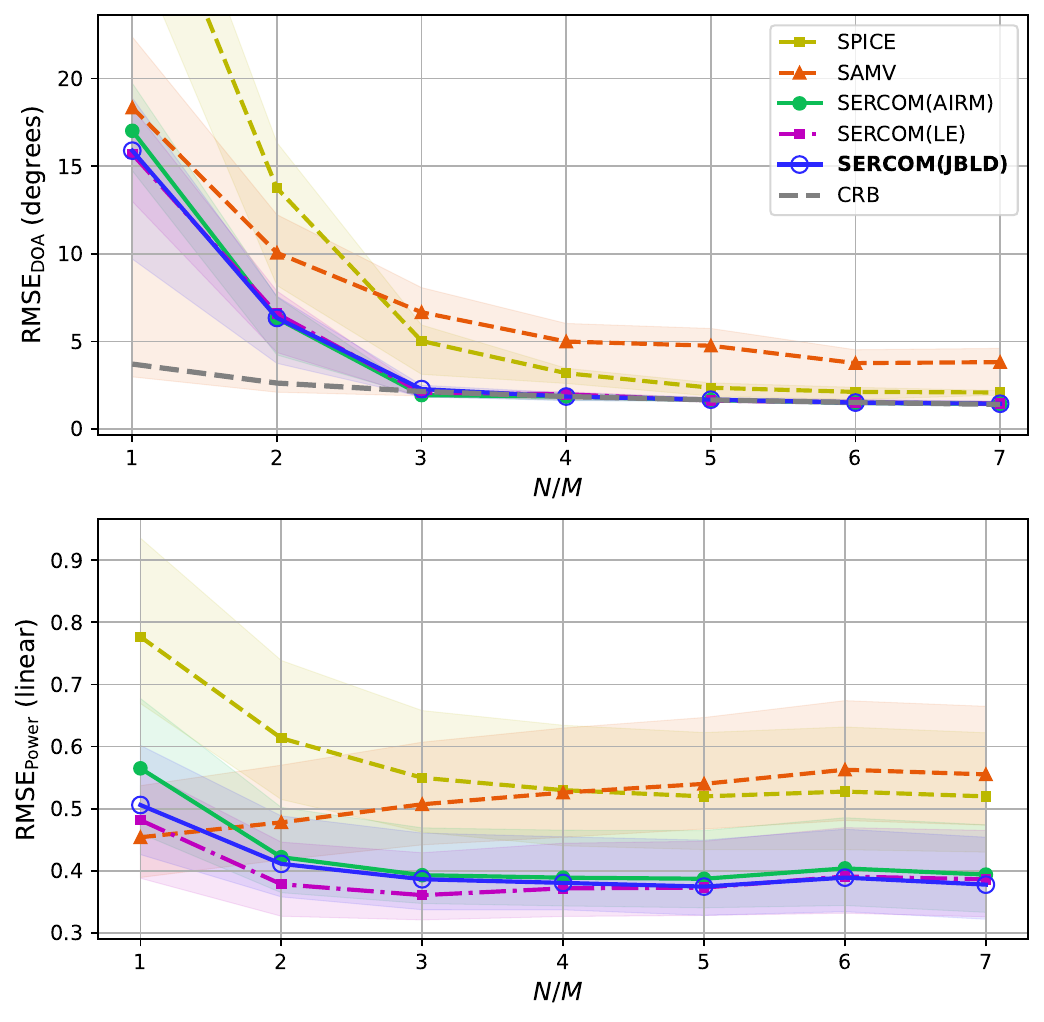}
    \caption{Algorithms' DOA (top) and power (bottom) estimation error while varying the number of snapshots, $N$ (normalized by number of sensors, $M=12$).}
    \label{fig:rmse_n_ula}
\end{figure}

\rev{
Figure~\ref{fig:rmse_n_ula} (top) depicts the RMSE of direction estimation as a function of the normalized number of snapshots, $N/M$. We observe that SERCOM variants achieve lower DOA RMSE than the Euclidean baselines for all values of $N/M$. SPICE performance is not shown at $N=M$ due to its very high error ($\ang{34}$), whereas the Riemannian approaches remain stable even in this sample-limited regime. SAMV appears less sensitive to the low-snapshot regime, but unlike SPICE it does not approach the CRB as rapidly when $N$ increases. The bottom panel shows the corresponding power RMSE, with Riemannian approaches again demonstrating accuracy and robustness to the number of snapshots. At low snapshot ratios $N/M$, SAMV often fails to activate the $-5$ dB component, with the third peak in its spectrum corresponding to spurious noise, and the estimated power concentrates on the two stronger sources, producing an artificially low power error. As $N/M$ increases, the weak source peak emerges and the power spreads across three lobes, causing the power error to increase.
}

\rev{
These results emphasize the challenge in the few-sample regime, where the sample covariance is a poor approximation of the true covariance. In this setting, we are outside the asymptotic regime of Theorem~\ref{thm:asymptotic_closenseness},
and, as expected, the criteria diverge. Empirically, the Euclidean-based methods suffer from degradation, while the Riemannian formulations remain robust and achieve significant performance gains in both direction and power estimation.
}

\rev{
\textbf{\textit{Off-grid sensitivity.}}
We next quantify sensitivity to grid mismatch. In this experiment only, we use a coarser grid with $1^\circ$ spacing to emphasize off-grid effects. We keep $M=12$ sensors and $N=50$ snapshots, and place two sources at directions $\theta_1=35^\circ+\Delta\theta$ and $\theta_2=51^\circ+\Delta\theta$. Both sources have power $0$\,dB and we set $\text{SNR}=0$\,dB. We sweep $\Delta\theta$ from $0^\circ$ to $0.5^\circ$. We also include ESPRIT as a representative gridless DOA baseline. We report DOA RMSE, focusing on direction bias induced by discretization mismatch.
}
\begin{figure}[htbp]
    \centering
    \includegraphics[width=\figW\columnwidth]{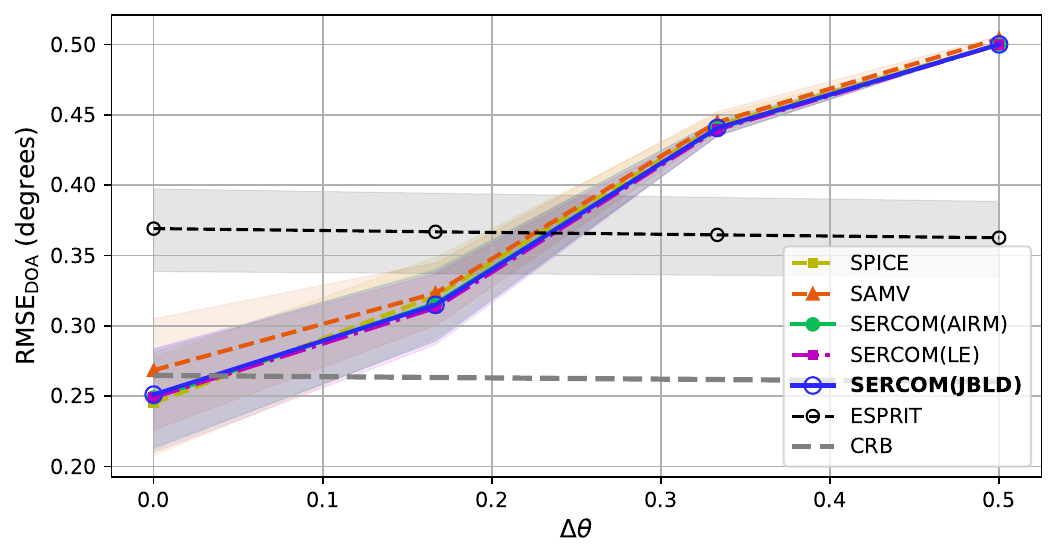}
    \caption{Algorithms' DOA estimation error for values of $\Delta\theta$ (ULA, 1$^\circ$ grid), including ESPRIT.}
    \label{fig:rmse_dtheta}
\end{figure}

\rev{
In Figure~\ref{fig:rmse_dtheta}, the DOA RMSE of all on-grid spectrum estimators (SPICE, SAMV, and SERCOM variants) increases as the sources move off the grid, reflecting discretization mismatch. At $\Delta\theta=0$ they achieve the lowest error\footnote{The CRB is a lower bound for unbiased estimators. The grid-based methods' sub-CRB performance at $\Delta\theta=0$ is a result of their inherent bias toward the grid points.}, but the error grows steadily with $\Delta\theta$, and for sufficiently large mismatch their RMSE exceeds that of ESPRIT. Of course, ESPRIT is not affected by the change in $\Delta\theta$ and its RMSE remains essentially constant.
}

\rev{
\textbf{\textit{Correlated sources.}}
The following experiment evaluates robustness to source correlation. We remain with $M=12$ sensors, $N=50$ snapshots, and $\text{SNR}=0$\,dB. Two sources are placed at $\theta_1 = \ang{35}$ and $\theta_2 = \ang{41}$, with equal powers at $0$\,dB. Unlike the previous experiments, the sources are now correlated. We characterize the correlation using a coefficient $\rho \in \mathbb{R}$, varying it from $\rho=0$ (uncorrelated sources) to $\rho=1$ (fully coherent sources).
The source covariance matrix is defined by
\begin{equation}
    \mathbf{\Sigma} = 
    \begin{bmatrix}
        \sigma_1^2 & \rho \sigma_1 \sigma_2 \\
        \rho \sigma_1 \sigma_2 & \sigma_2^2
    \end{bmatrix}.
\end{equation}
We include ESPRIT here as well to illustrate the effect of correlation on a subspace-based gridless baseline.
}
\begin{figure}[htbp]
    \centering
    \includegraphics[width=\figW\columnwidth]{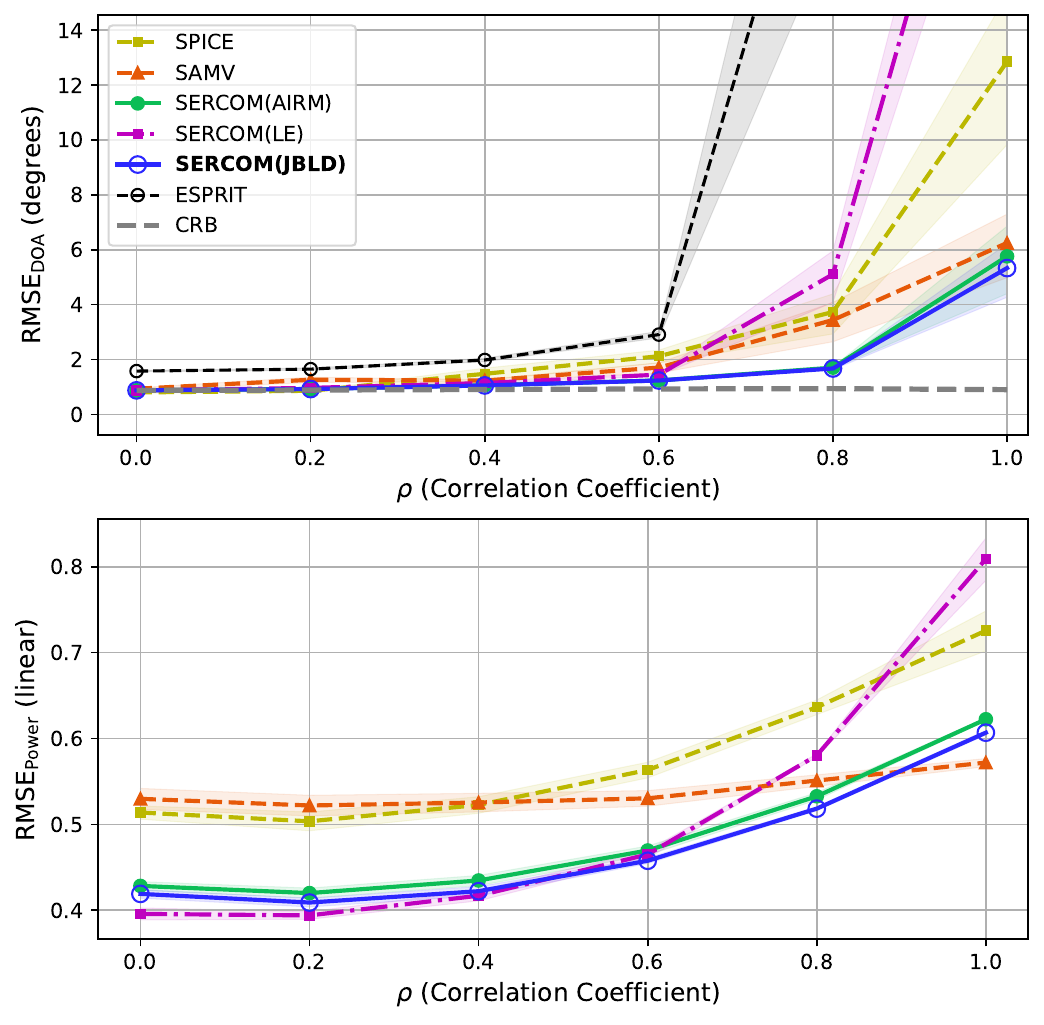}
    \caption{Algorithms' DOA (top) and power (bottom) estimation error as a function of the correlation coefficient, $\rho$.}
    \label{fig:rmse_rho}
\end{figure}

\rev{
Figure~\ref{fig:rmse_rho} (top) presents the DOA RMSE as a function of the correlation coefficient $\rho$. As $\rho$ increases, the estimation task becomes more challenging and the DOA RMSE rises for all methods. In particular, ESPRIT degrades sharply for highly correlated sources, reflecting the well-known sensitivity of subspace methods in this regime. SPICE also exhibits sensitivity to coherent sources, consistent with earlier reports in \cite{abeida2012iterative}, and the large errors of SERCOM(LE) are attributable to numerical instability in its optimization. Among all examined approaches, SERCOM(JBLD) remains the most stable and achieves the lowest DOA RMSE at high correlation. The bottom panel shows the corresponding power RMSE, where SERCOM(AIRM) and SERCOM(JBLD) remain stable across most values of $\rho$; at $\rho=1$, SAMV attains a slightly lower power RMSE than the SERCOM variants.
}

\rev{
These results highlight the challenge in the highly correlated regime. Source correlation introduces model mismatch for all algorithms, such that $\|\bbar{\R}\R^{-1}-\I\|_2$ may become substantially larger than zero, and the criteria diverge. Empirically, the performance of Euclidean-based criteria degrades rapidly in this setting, which reflects the fact that they neglect the underlying geometry. In contrast, the Riemannian-aware distances preserve robustness, allowing SERCOM(JBLD) and SERCOM(AIRM) to sustain reliable direction and power estimation even when sources are highly correlated.
}

\rev{
\textbf{\textit{UCA configuration.}}
Here, we demonstrate that the covariance-matching comparison is not restricted to a ULA. We repeat the SNR sweep using a semi-circular UCA configuration with $M=14$ sensors uniformly distributed along a semicircle, with the radius configured to maintain a half-wavelength inter-element arc length. We simulate three sources at directions $\theta_1 = \ang{35}$, $\theta_2 = \ang{43}$, and $\theta_3 = \ang{51}$. All sources lie in the array plane; hence, only azimuthal DOAs are estimated. The power of the first two sources is $0$~dB, and the power of the third source is $-5$~dB. The number of snapshots is $N = 50$.
}

\begin{figure}[htbp]
    \centering
    \includegraphics[width=\figW\columnwidth]{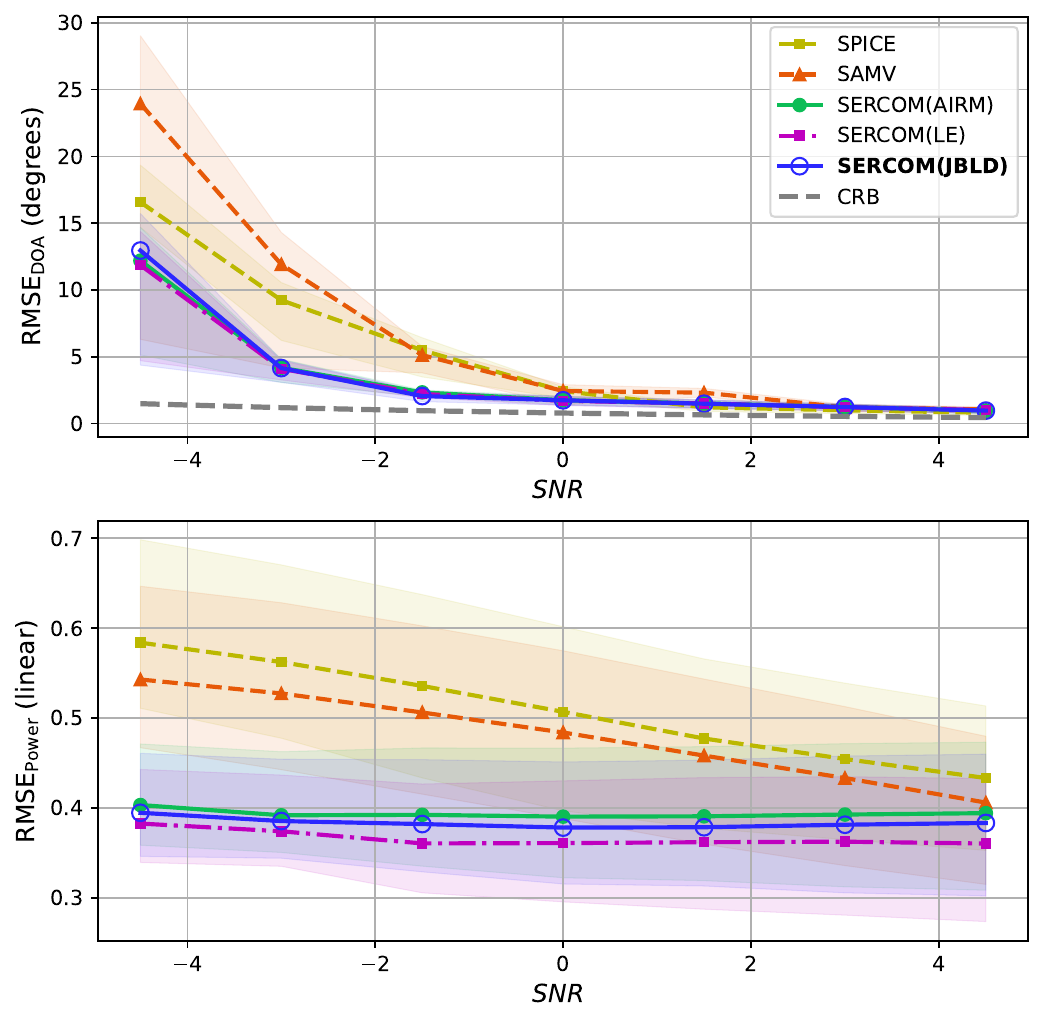}
    \caption{Algorithms' DOA (top) and power (bottom) estimation error at different SNR values, on a UCA configuration.}
    \label{fig:rmse_snr_uca}
\end{figure}
\rev{
Figure~\ref{fig:rmse_snr_uca} reports RMSE versus SNR for the half-circle UCA. The overall trends are consistent with the ULA results, demonstrating that the compared covariance-matching methods apply beyond the ULA configuration.
}

\rev{
\textbf{\textit{Runtime.}} 
Finally, Fig.~\ref{fig:runtime} shows a box plot of the runtime (in seconds) for all algorithms with $M=12$ and $M=120$ sensors, based on $L=500$ Monte Carlo runs. The results illustrate how the computational cost scales with $M$. SAMV and SPICE exhibit higher runtimes, since they require a matrix inversion or matrix square root in each iteration. The AIRM- and LE-based variants become prohibitively slow for larger arrays because they require an eigen-decomposition in each iteration. In contrast, our proposed method, SERCOM(JBLD), avoids this overhead and maintains favorable scaling with $M$.
}
\begin{figure}[htbp]
    \centering
    \includegraphics[width=\figW\columnwidth]{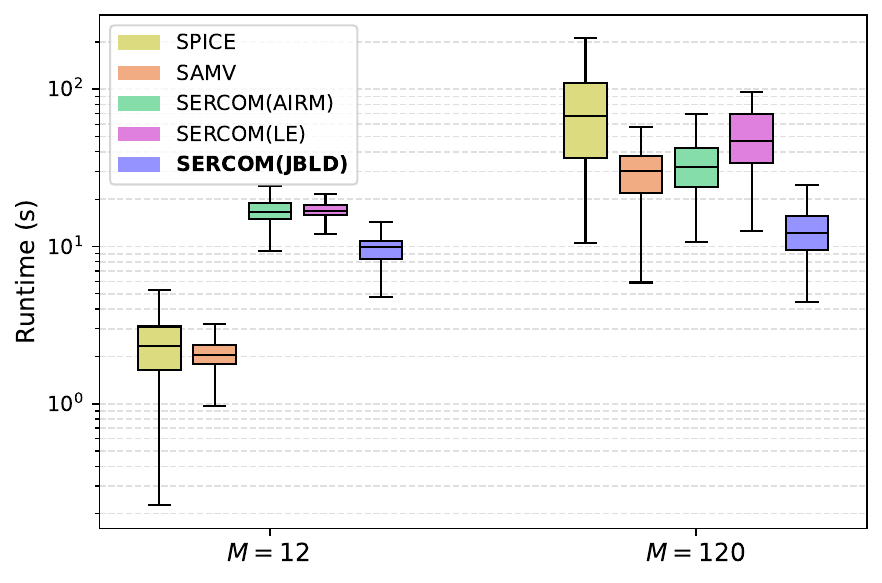}
    \caption{Algorithms' runtime (in seconds) for an array with $M$ sensors.}
    \label{fig:runtime}
\end{figure}
\rev{
\begin{table}[htbp]
\centering
\caption{Average number of iterations.}
\label{tab:num_iters}
\begin{tabular}{lcc}
\hline
\textbf{Algorithm} & \textbf{$M=12$} & \textbf{$M=120$} \\
\hline
SPICE & 3447.71 & 2928.82 \\
SAMV  & 3630.79 & 2635.85 \\
SERCOM(AIRM)  & 4546.20 & 2109.64 \\
SERCOM(LE)    & 4631.62 & 3464.05 \\
SERCOM(JBLD)  & 4389.59 & 1866.46 \\
\hline
\end{tabular}
\end{table}
}

\section{CONCLUSION} \label{sec:conclusion}
We introduced SERCOM, a novel spatial power estimation method based on Riemannian covariance matching, applied to DOA estimation. We established a connection between widely used Euclidean criteria and HPD geometry-aware distances, which motivated the use of the computationally efficient JBLD divergence as a criterion for covariance matching on the manifold. \rev{The proposed method achieves improved direction and power estimation relative to SPICE and SAMV, particularly at low SNR, with limited snapshots, and for correlated sources. While the RMSE performance is comparable across the different Riemannian metrics tested, our JBLD-based approach demonstrates a slight advantage for correlated sources and a significant reduction in runtime}. \rev{These results suggest that SERCOM should be considered a robust candidate for evaluation in practical, real-world settings where signal conditions are non-ideal.}

Future work will explore extending the proposed framework to \rev{gridless} DOA estimation, aiming to combine the robustness of Riemannian metrics with the resolution benefits of continuous-domain formulations.

\clearpage \clearpage
\bibliographystyle{IEEEtran}
\bibliography{references} 

\clearpage
\appendices

\onecolumn
\section{Derivation of $\psi$ functions} \label{appendix:derivation_psi_functions}
\noindent\textit{The AIRM criterion.} \\
\indent We would like to show that $\psi_{\mathrm{AIRM}}(\lambda)=(\log \lambda)^2$. From $\mathcal{D}_{\mathrm{AIRM}}^2(\R,\widehat{\R}) = \|\log(\Q)\|_F^2$, this is immediate.

\bigskip\noindent\textit{The AMV criterion.} \\
\indent We would like to show that $\psi_{\mathrm{AMV}}(\lambda)=(\lambda-1)^2$. From $\mathcal{D}_{\mathrm{AMV}}^2(\R,\widehat{\R}) = \|\Q - \I\|_F^2$, this is immediate.

\bigskip\noindent\textit{The SPICE criterion.} \\
\indent We would like to show that $\psi_{\mathrm{SPICE}}(\lambda)=\Big(\sqrt{\lambda}-\tfrac{1}{\sqrt{\lambda}}\Big)^2$.
Notice that
\begin{equation}
\begin{aligned}
\label{eq:spice-trace-form}
&\mathcal{D}_{\mathrm{SPICE}}^2(\R,\widehat{\R}) &\;=\;& \|\invsqrtm{\R}(\R - \widehat\R)\invsqrtm{\widehat\R}\|_F^2 \\
&\quad &\;=\;& \mathrm{tr}\left(\R^{-1}(\R-\widehat\R)\widehat\R^{-1}(\R-\widehat\R)\right)\\
&\quad &\;=\;& \mathrm{tr}\left((\I-\R^{-1}\widehat\R)(\widehat\R^{-1}\R-\I)\right)\\
&\quad &\;=\;& \mathrm{tr}\!\big(\widehat\R^{-1}\R\big) + \mathrm{tr}\!\big(\R^{-1}\widehat\R\big) - 2M .
\end{aligned}
\end{equation}
Since $\mathrm{tr}(\R^{-1}\widehat\R)=\mathrm{tr}(\Q)$ and
$\mathrm{tr}(\widehat\R^{-1}\R)=\mathrm{tr}(\Q^{-1})$, we can write \eqref{eq:spice-trace-form} by
\begin{equation}
\mathcal{D}_{\mathrm{SPICE}}^2
= \mathrm{tr}(\Q)+\mathrm{tr}(\Q^{-1})-2M
= \sum_{m=1}^M \big(\lambda_m + \lambda_m^{-1} - 2\big).
\end{equation}
Finally, for each $\lambda>0$,
\[
\lambda+\lambda^{-1}-2
= \left(\sqrt{\lambda}-\frac{1}{\sqrt{\lambda}}\right)^2 ,
\]
hence the per-eigenvalue contribution is
\[
\psi_{\mathrm{SPICE}}(\lambda)
= \left(\sqrt{\lambda}-\frac{1}{\sqrt{\lambda}}\right)^2 .
\]

\bigskip\noindent\textit{The JBLD criterion.} \\
\indent We would like to show that $\psi_{\mathrm{JBLD}}(\lambda)=\log\!\Big(\tfrac{1+\lambda}{2\sqrt{\lambda}}\Big)$. Observe that 
\begin{equation}
\begin{aligned}
 &\mathcal{D}^2_{\mathrm{JBLD}}(\R,\widehat\R) &\;=\;& \log\left|\frac{\R + \widehat\R}{2}\right| - \frac{1}{2}\log|\R \widehat\R| \\
 &\quad &\;=\;& \log\left|\R + \widehat\R\right| - \frac{1}{2}\log|\R \widehat\R| - M\log(2)\\
 &\quad &\;=\;&  \log\left|\I + \widehat\R\R^{-1}\right| + \log|\R| - \frac{1}{2}\log|\widehat\R\R^{-1}| -\frac{1}{2}\log|\R^2| - M\log(2)\\
 &\quad &\;=\;& \log\left|\I + \widehat\R\R^{-1}\right| - \frac{1}{2}\log|\widehat\R\R^{-1}| - M\log(2).\\
\end{aligned}
\end{equation}
Now, using the fact that determinant is invariant under matrix similarity, and  $\widehat\R\R^{-1}$ is similar to $\Q$, we can write:
\begin{equation}
\begin{aligned}
&\mathcal{D}^2_{\mathrm{JBLD}}(\R,\widehat\R) &\;=\;& \sum_{m=1}^M{\left(\log(1+\lambda_m) - \tfrac12\log(\lambda_m) - \log(2) \right)}\\
&\quad &\;=\;& \sum_{m=1}^M{\log\left(\frac{1+\lambda_m}{2\sqrt{\lambda_m}}\right)},\\
\end{aligned}
\end{equation}
which concludes the derivation of $\psi$ function for all criteria.

\clearpage
\section{Proof of Theorem 1 (Asymptotic equivalence of criteria)} \label{appendix:proof_theorem1}

\noindent  Let us first bound $\|\widehat{\R}\R^{-1}-\I\|_2$ by a sum of two terms: model mismatch and sampling error.

\indent Write
\begin{equation}
\label{eq:model_plus_sampling}
\|\widehat{\R}\R^{-1}-\I\|_2 \;\le\; \|\bbar\R \R^{-1}-\I\|_2 \;+\; \|(\widehat\R-\bbar\R)\R^{-1}\|_2 .
\end{equation}
Denote the r.h.s of \eqref{eq:model_plus_sampling} by $\varepsilon$. For $\R\in\mathcal B_{\mathrm{AIRM}}(\bbar\R,\rho)$, the generalized eigenvalues of $(\R,\bbar\R)$ lie in $[e^{-\rho},e^\rho]$, hence
\begin{equation}
\label{eq:bounded_by_ball}
\|\bbar\R \R^{-1}-\I\|_2=\|\S^{-1}-\I\|_2\le e^\rho-1,
\end{equation}
\begin{equation}
\label{eq:norm_invR}
\|\R^{-1}\|_2\le \|\bbar\R^{-1}\|_2\,\|\S^{-1}\|_2\le \|\bbar\R^{-1}\|_2\,e^\rho,
\end{equation}
where $\S:=\bbar\R^{-1/2} \R\,\bbar\R^{-1/2}$.
For Gaussian samples, the spectral concentration inequality (\cite{vershynin2018high}, see Theorem 4.7.1 or Exercise 4.7.3) yields, with probability $\ge 1-\delta$,
\begin{equation}
\label{eq:by_vershynin}
\|\widehat\R-\bbar\R\|_2 \ \le\ \|\bbar\R\|_2\Big(\sqrt{\tfrac{M+\log(2/\delta)}{N}}+\tfrac{M+\log(2/\delta)}{N}\Big).
\end{equation}
Expanding \eqref{eq:model_plus_sampling}, we get:
\begin{equation}
\begin{aligned}
\label{eq:epsilon_ineq_N}
&\varepsilon &\;=\;& \|\bbar\R \R^{-1}-\I\|_2 \;+\; \|(\widehat\R-\bbar\R)\R^{-1}\|_2 \\ 
&\quad &\;\le\;& \|\bbar\R \R^{-1}-\I\|_2 \;+\; \|(\widehat\R-\bbar\R)\|_2\|\R^{-1}\|_2 \\
&\quad &\;\le\;&
\ e^\rho-1\;+\;\|(\widehat\R-\bbar\R)\|_2\|\bbar\R^{-1}\|_2\,e^\rho \\
&\quad &\;\le\;&
\ e^\rho-1\;+\;\kappa(\bbar\R)\,e^\rho\,\sqrt{\tfrac{M+\log(2/\delta)}{N}}
\end{aligned}
\end{equation}
where we used the submultiplicativity of the operator norm for the first inequality, \eqref{eq:bounded_by_ball} and \eqref{eq:norm_invR} for the second inequality, and \eqref{eq:by_vershynin} by discarding the lower-order linear term (negligible when $N \gg M$) in the last inequality.

Solving \eqref{eq:epsilon_ineq_N} for $N$ yields the $N_0(\varepsilon,\delta,\rho)$ stated in the theorem.

\bigskip\noindent Next, we use Taylor expansions and bound the absolute difference between the criteria.

\indent Let $x = \lambda - 1$ with $|x|\le\varepsilon<1$. Note that:
\begin{equation}
\psi_{\mathrm{AMV}}(x) = x^2,
\end{equation}
And that
\begin{equation}
\begin{aligned}
&\psi_{\mathrm{SPICE}}(x) &\;=\;& \Big(\sqrt{1+x}-\frac{1}{\sqrt{1+x}}\Big)^{\!2}
= \frac{x^2}{1+x} \\
&\quad &\;=\;&x^2 - \frac{x^3}{x+1} \\
&\quad&\;=\;&x^2 - x^3 +\frac{x^4}{1+x}.
\end{aligned}
\end{equation}

The 3-rd order Taylor expansion around zero of $\psi_{\mathrm{AIRM}}(x)$ is by:
\begin{equation}
    \psi_{\mathrm{AIRM}}(x)=\log^2(1+x) = x^2 -x^3 + r^{(\mathrm{AIRM})}_3(x), \label{eq:taylor_airm_3rd}
\end{equation}
with the remainder defined for some $\xi\in(0,x)$ by:
\begin{equation}
r^{(\mathrm{AIRM})}_3(x)= x^4 \cdot \frac{11 - 6\log(1+\xi)}{12(1+\xi)^4},
\end{equation}
hence: 
\begin{equation}    
|r^{(\mathrm{AIRM})}_3(x)| \le \varepsilon^4\cdot\frac{11+6|\log(1-\varepsilon)|}{12(1-\varepsilon)^4}.
\end{equation}    

To align the leading quadratic term with the other criteria, we develop the 3-rd order Taylor expansion around zero of $8\psi_{\mathrm{JBLD}}(x)$. It is expressed by:
\begin{equation}
\begin{aligned}
&8\psi_{\mathrm{JBLD}}(x) &\;=\;& 8\log(1+\tfrac{x}{2})-4\log(1+x) \\
& &\;=\;& x^2 - x^3 + r_3^{(\mathrm{JBLD})}(x)
\end{aligned}
\end{equation}
with the remainder defined for some $\xi\in(0,x)$ by:
\begin{equation}
r_3^{(\mathrm{JBLD})}(x) = x^4 \cdot \left(\frac{1}{(\xi+1)^4} - \frac{1}{(\xi+2)^4} \right)
\end{equation}
hence: 
\begin{equation}    
|r_3^{(\mathrm{JBLD})}(x)| \le \varepsilon^4 \cdot \left(\frac{1}{(1-\varepsilon)^4} + \frac{1}{(2-\varepsilon)^4} \right)
\end{equation}

We now derive the per-eigenvalue bounds and then sum over $m$. Each of $\big\{\mathrm{AMV},\mathrm{AIRM},\mathrm{JBLD}\big\}$ will be compared to $\mathrm{SPICE}$ criterion.

\bigskip\noindent\underline{(i) AMV vs.\ SPICE:}\\
\indent Compute
\begin{equation*}
\begin{aligned}
&|\psi_{\mathrm{AMV}}(x) - \psi_{\mathrm{SPICE}}(x)| = |\frac{x^3}{x+1}| \le \varepsilon^3\cdot\frac{1}{1-\varepsilon}.
\end{aligned}
\end{equation*}
Define $C_{\mathrm{AMV}}(\varepsilon) =\frac{1}{1-\varepsilon}$, which converges to a finite constant when $\varepsilon \to 0$. 
Now, bounding criteria absolute difference using the triangle inequality:
\begin{equation}
\begin{aligned}
&\big|\mathcal{D}_{\mathrm{AMV}}^2-\mathcal{D}_{\mathrm{SPICE}}^2\big| \ &\;\le\;& \sum_m{|\psi_{\mathrm{AMV}}(\lambda_m) - \psi_{\mathrm{SPICE}}(\lambda_m)|}\\ 
&\ &\;\le\;&  M\cdot C_{\mathrm{AMV}}(\varepsilon) \cdot \varepsilon^3
\end{aligned}
\end{equation}

\bigskip\noindent\underline{(ii) AIRM vs.\ SPICE:}\\
\indent Compute
\begin{equation*}
\begin{aligned}
&|\psi_{\mathrm{AIRM}}(x) - \psi_{\mathrm{SPICE}}(x)| &\;=\;& |r^{(\mathrm{AIRM})}_3(x) - \frac{x^4}{x+1}| \\
&\quad &\;\le\;& \varepsilon^4 \cdot \left( \frac{11+6|\log(1-\varepsilon)|}{12(1-\varepsilon)^4} + \frac{1}{1-\varepsilon}\right)
\end{aligned}
\end{equation*}
Define $C_{\mathrm{AIRM}}(\varepsilon) =\left( \frac{11+6|\log(1-\varepsilon)|}{12(1-\varepsilon)^4} + \frac{1}{1-\varepsilon}\right)$, which converges to a finite constant when $\varepsilon \to 0$. 
Now, bounding criteria absolute difference using the triangle inequality:
\begin{equation}
\begin{aligned}
&\big|\mathcal{D}_{\mathrm{AIRM}}^2-\mathcal{D}_{\mathrm{SPICE}}^2\big| \ &\;\le\;& \sum_m{|\psi_{\mathrm{AIRM}}(\lambda_m) - \psi_{\mathrm{SPICE}}(\lambda_m)|}\\ 
&\ &\;\le\;&  M\cdot C_{\mathrm{AIRM}}(\varepsilon) \cdot \varepsilon^4
\end{aligned}
\end{equation}

\bigskip\noindent\underline{(iii) JBLD vs.\ SPICE:}\\
\indent Compute
\begin{equation*}
\begin{aligned}
&|8\psi_{\mathrm{JBLD}}(x) - \psi_{\mathrm{SPICE}}(x)| &\;=\;& |r^{(\mathrm{JBLD})}_3(x) - \frac{x^4}{x+1}| \\
&\quad &\;\le\;& \varepsilon^4 \cdot \left( \frac{1}{(1-\varepsilon)^4} + \frac{1}{(2-\varepsilon)^4} + \frac{1}{1-\varepsilon}\right)
\end{aligned}
\end{equation*}

Define $C_{\mathrm{JBLD}}(\varepsilon) = \left( \frac{1}{(1-\varepsilon)^4} + \frac{1}{(2-\varepsilon)^4} + \frac{1}{1-\varepsilon}\right)$, which converges to a finite constant when $\varepsilon \to 0$. 
Now, bounding criteria absolute difference using the triangle inequality:
\begin{equation}
\begin{aligned}
&\big|8\mathcal{D}_{\mathrm{JBLD}}^2-\mathcal{D}_{\mathrm{SPICE}}^2\big| \ &\;\le\;& \sum_m{|\psi_{\mathrm{JBLD}}(\lambda_m) - \psi_{\mathrm{SPICE}}(\lambda_m)|}\\ 
&\ &\;\le\;&  M\cdot C_{\mathrm{JBLD}}(\varepsilon) \cdot \varepsilon^4
\end{aligned}
\end{equation}

Combining the bounds (i)–(iii) yields the stated inequalities, which completes the proof.
\qed

\clearpage
\section{Proof of Theorem 2 (Criteria robustness to eigenvalue deviations)} \label{appendix:proof_theorem2}
To exploit monotonicity, we work in the log–eigenvalue coordinates:
\[
u_m \triangleq \tfrac12 \log \lambda_m
\]
so that $\lambda_m=e^{2u_m}$.
The scalar penalties in \eqref{eq:psi_functions} can be written as functions of $u$:
\begin{equation}
\begin{aligned}    
&\psi_{\mathrm{AMV}}(\lambda)= (e^{2u}-1)^2, \;
&\psi_{\mathrm{SPICE}}(\lambda)= 4\sinh^2 u,\\
&\psi_{\mathrm{AIRM}}(\lambda)= 4u^2, \;
&\psi_{\mathrm{JBLD}}(\lambda)= \log(\cosh u).
\end{aligned}
\end{equation}

Denote $U_{\max} \triangleq  \tfrac12\max\{\log(1+\varepsilon),-\log(1-\varepsilon)\}$. Notice that by theorem assumptions, for $m \ne r$:
\[
|u_m| \le U_{\max} \le u_r.
\]

\noindent Now, we derive a useful equivalence that will aid the proof. 

Set for each criterion:
\begin{equation}
S^{(\bullet)}\!\triangleq\!\sum_{m\neq r}\psi_{\bullet}(\lambda_m). 
\end{equation}

Observe that $\psi_{\bullet}(\lambda) \ge 0$ for all criteria, which can be verified algebraically and is illustrated in Figure~\ref{fig:psi_functions}. Consequently, $S^{(\bullet)} \ge 0$ holds for all criteria.

For any two criteria $A$ and $B$,
$
s_r^{(A)} \le s_r^{(B)}
\quad\Longleftrightarrow\quad
\frac{\psi_{A}(\lambda_r)}{S^{(A)}} \le \frac{\psi_{B}(\lambda_r)}{S^{(B)}} .
$
Equivalently,
\begin{equation}\label{eq:equiv-star}
s_r^{(A)} \le s_r^{(B)}
\quad\Longleftrightarrow\quad
\frac{S^{(B)}}{S^{(A)}} \le \frac{\psi_B(\lambda_r)}{\psi_A(\lambda_r)} .
\end{equation}
with the convention that if $S^{(A)}=S^{(B)}=0$ then $s_r^{(A)}=s_r^{(B)}=1$ and the inequality holds with equality.

\bigskip
\noindent Next, we show twp monotonicity lemmas on the log–scale.

Observe the following two lemmas.
\begin{enumerate}
\item[(L1)] For $u\in\mathbb{R}$, define the ratio:
\[
\displaystyle
\alpha(u)\triangleq \frac{\psi_{\mathrm{SPICE}}(\lambda)}{\psi_{\mathrm{AIRM}}(\lambda)}
=\Big(\frac{\sinh u}{u}\Big)^{\!2}, \; \text{with $\alpha(0)\triangleq1$.}
\]
Then $\alpha(u)$ is strictly increasing in $|u|$.

\smallskip\noindent\textit{Proof.} \\
For $u>0$, define $f(u)=\sinh u/u$. Then,
$f'(u)=\frac{u\cosh u-\sinh u}{u^2}$.
The numerator $g(u)=u\cosh u-\sinh u$ satisfies $g'(u)=u\sinh u>0$ and $g(0)=0$,
so $g(u)>0$ for $u>0$. Hence $f'(u)>0$ and $\alpha(u)=f(u)^2$ increases on $u>0$. Notice that $\alpha(u)$ is an even function, hence $\alpha(u)$ increases in $|u|$.

\item[(L2)] For $u\in\mathbb{R}$, define the ratio:
\[
\beta(u)\triangleq \frac{\psi_{\mathrm{JBLD}}(\lambda)}{\psi_{\mathrm{AIRM}}(\lambda)}
=\frac{\log(\cosh u)}{4u^2}, \; \text{with $\beta(0)\triangleq1/8$.}
\]
Then $\beta(u)$is strictly decreasing in $|u|$.

\smallskip\noindent\textit{Proof.} \\
For $u>0$,
$\beta'(u)=\frac{u\tanh u-2\log(\cosh u)}{4u^3}$.
Define the numerator $h(u) \triangleq u\tanh u-2\log(\cosh u)$. We got $h'(u) = u\text{sech}^2 u - \tanh u$ and $h''(u) = -2u\cdot\text{sech}^2u\cdot\tanh u$.
For $u>0$, $h''(u)$ is negative, and $h'(0) = 0$, it follows that $h'(u) < 0$. Hence, for $u>0$, $h(u)$ is strictly decreasing with $h(0)=0$, so $h(u)<0$ and thus $\beta'(u) < 0$. Notice that $\beta(u)$ is an even function, hence $\beta$ decreases in $|u|$.
\end{enumerate}

\bigskip\noindent We are now ready to prove the following inequalities

\bigskip\noindent\underline{(i) $s_r^{(\mathrm{AMV})} \ge s_r^{(\mathrm{SPICE})}$:}\\
\indent Notice that
\[
\psi_{\mathrm{AMV}}(\lambda_m)=(\lambda_m-1)^2=\lambda_m(\sqrt\lambda_m-\frac{1}{\sqrt\lambda_m})^2 = \lambda_m\,\psi_{\mathrm{SPICE}}(\lambda_m),
\]
for all $m$. Hence
\[
\frac{S^{(\mathrm{AMV})}}{S^{(\mathrm{SPICE})}}
=\frac{\sum_{m\neq r}\lambda_m\,\psi_{\mathrm{SPICE}}(\lambda_m)}{\sum_{m\neq r}\psi_{\mathrm{SPICE}}(\lambda_m)}
\le \max_{m\neq r}\lambda_m \le \lambda_r,
\]
since, $\max_{m\neq r}\lambda_m \le 1+\varepsilon \le 1+\Delta = \lambda_r$.

Therefore,
\[
\frac{S^{(\mathrm{AMV})}}{S^{(\mathrm{SPICE})}} \le \lambda_r = \frac{\psi_{\mathrm{AMV}}(\lambda_r)}{\psi_{\mathrm{SPICE}}(\lambda_r)},
\]
and by \eqref{eq:equiv-star} (with $B=\mathrm{AMV}$, $A=\mathrm{SPICE}$)
we conclude $s_r^{(\mathrm{AMV})}\ge s_r^{(\mathrm{SPICE})}$.

\bigskip\noindent\underline{(ii) $s_r^{(\mathrm{SPICE})} \ge s_r^{(\mathrm{AIRM})}$:}\\
\indent For each $m\neq r$, by (L1), evenness of $\alpha(u)$ and $ |u_m|\le U_{\max}\le u_r$:
\[
\frac{\psi_{\mathrm{SPICE}}(\lambda_m)}{\psi_{\mathrm{AIRM}}(\lambda_m)}
= \alpha(u_m) = \alpha(|u_m|) \le \alpha(u_r)
= \frac{\psi_{\mathrm{SPICE}}(\lambda_r)}{\psi_{\mathrm{AIRM}}(\lambda_r)}.
\]
Multiplying by $\psi_{\mathrm{AIRM}}(\lambda_m)\ge 0$ and summing over $m\neq r$ yields
\[
S^{(\mathrm{SPICE})} \le
\frac{\psi_{\mathrm{SPICE}}(\lambda_r)}{\psi_{\mathrm{AIRM}}(\lambda_r)}\,S^{(\mathrm{AIRM})}.
\]
Equivalently (for non-degenerate case where $S^{(\mathrm{AIRM})}>0$:
\[
\dfrac{S^{(\mathrm{SPICE})}}{S^{(\mathrm{AIRM})}} \le
\dfrac{\psi_{\mathrm{SPICE}}(\lambda_r)}{\psi_{\mathrm{AIRM}}(\lambda_r)}.
\]
Applying \eqref{eq:equiv-star} (with $B=\mathrm{SPICE}$, $A=\mathrm{AIRM}$) gives
$s_r^{(\mathrm{SPICE})}\ge s_r^{(\mathrm{AIRM})}$.

\bigskip\noindent\underline{(iii) $s_r^{(\mathrm{AIRM})} \ge s_r^{(\mathrm{JBLD})}$:}\\
\indent For each $m\neq r$, by (L2), evenness of $\beta(u)$ and $ |u_m|\le U_{\max}\le u_r$:
\[
\frac{\psi_{\mathrm{JBLD}}(\lambda_m)}{\psi_{\mathrm{AIRM}}(\lambda_m)}
= \beta(u_m) = \beta(|u_m|) \ge \beta(u_r)
= \frac{\psi_{\mathrm{JBLD}}(\lambda_r)}{\psi_{\mathrm{AIRM}}(\lambda_r)}.
\]
Multiplying by $\psi_{\mathrm{AIRM}}(\lambda_m)\ge 0$ and summing over $m\neq r$ gives
\[
S^{(\mathrm{JBLD})} \ge
\frac{\psi_{\mathrm{JBLD}}(\lambda_r)}{\psi_{\mathrm{AIRM}}(\lambda_r)}\,S^{(\mathrm{AIRM})}.
\]
Equivalently,
\[
\dfrac{S^{(\mathrm{JBLD})}}{S^{(\mathrm{AIRM})}} \ge
\dfrac{\psi_{\mathrm{JBLD}}(\lambda_r)}{\psi_{\mathrm{AIRM}}(\lambda_r)}.
\]
Applying \eqref{eq:equiv-star} (with $B=\mathrm{AIRM}$, $A=\mathrm{JBLD}$) yields
$s_r^{(\mathrm{AIRM})}\ge s_r^{(\mathrm{JBLD})}$.

Combining (i),(ii) and (iii) concludes the proof.
\qed

\end{document}